\documentclass[journal]{IEEEtran}

\ifCLASSINFOpdf
\else
\fi

\usepackage[cmex10]{amsmath}

\usepackage{enumitem}
\usepackage{graphicx}
\usepackage{color}
\usepackage{amsmath}
\usepackage{mathtools}
\usepackage{multicol}
\usepackage{multirow}
\usepackage[english]{babel}
\usepackage{blindtext}
\usepackage{algorithm}
\usepackage{algorithmic}
\usepackage{balance}
\usepackage{amsfonts}
\usepackage{bm}
\usepackage{stfloats}
\usepackage{subfig}
\usepackage{amsthm}
\usepackage{amssymb}
\usepackage{setspace}
\usepackage[nosort]{cite}
\usepackage{CJK}
\usepackage{cite}
\usepackage{caption}
\usepackage{comment}
\usepackage{array,multirow}
\usepackage{graphicx}

\newtheorem{theorem}{Theorem}

\newtheorem{proposition}{Proposition}
\theoremstyle{plain}

\usepackage[table]{xcolor}
\usepackage{amsmath}
\allowdisplaybreaks[4]

\begin{document}
%

\title{Channel Modeling and Rate Analysis of Optical Inter-Satellite Link (OISL)}
%
%
%

\author{\IEEEauthorblockN{Bodong Shang,~\IEEEmembership{Member,~IEEE}}, Shuo Zhang, and Zi Jing Wong
\thanks{B. Shang, S. Zhang, and Z J. Wong are with Eastern Institute for Advanced Study, Eastern Institute of Technology, Ningbo 315200, China (E-mails: bdshang@eitech.edu.cn, shuozhang@eitech.edu.cn, zijing@eitech.edu.cn).}}

\maketitle

\begin{abstract}
Optical inter-satellite links (OISLs) improve connectivity between satellites in space. They offer advantages such as high-throughput data transfer and reduced size, weight, and power requirements compared to traditional radio frequency transmission. However, the channel model and communication performance for long-distance inter-satellite laser transmission still require in-depth study. 
In this paper, we first develop a channel model for OISL communication within non-terrestrial networks (NTN) by accounting for pointing errors caused by satellite jitter and tracking noise. We derive the distributions of the channel state arising from these pointing errors and calculate their average value.
Additionally, we determine the average achievable data rate for OISL communication in NTN and design a cooperative OISL system, highlighting a trade-off between concentrating beam energy and balancing misalignment. We calculate the minimum number of satellites required in cooperative OISLs to achieve a targeted data transmission size while adhering to latency constraints. This involves exploring the balance between the increased data rate of each link and the cumulative latency across all links.
Finally, simulation results validate the effectiveness of the proposed analytical model and provide insights into the optimal number of satellites needed for cooperative OISLs and the optimal laser frequency to use.
\end{abstract}


\IEEEpeerreviewmaketitle

\section{Introduction}

\IEEEPARstart{N}{on}-terrestrial network (NTN), particularly those using low-Earth orbit (LEO) satellites, provide widespread wireless connectivity through radio frequency (RF) signal transmissions. It’s important to note that these satellites are expected to be interconnected to transfer data to designated ground stations for internet access. Typically, ground stations are located in fixed and constrained areas. Therefore, establishing inter-satellite links (ISLs) among satellites is crucial for various applications within NTNs.

Optical ISLs (OISLs) are developed using laser communications to achieve very high-throughput data transfer \cite{10462050}. Unlike the congested wireless RF spectrum, which includes bands such as the S-band, Ka-band, and Ku-band, the infrared portion of the electromagnetic spectrum used in OISL—ranging from approximately 300 gigahertz (GHz) to around 430 terahertz (THz)—provides an exceptional bandwidth. This characteristic enhances the potential for encoding more data into the waveform. In comparison to RF links, OISL can gather more energy, allowing for a reduction in the size, weight, and power requirements of the laser transmitter and detector \cite{9855659}.

In contrast to laser communication between a satellite and a ground station, the propagation of the OISL signal is primarily influenced by pointing errors. These errors can significantly reduce the power of the received signal at the detector, underscoring the importance of addressing this issue. Pointing errors in OISL arise from the satellite's jitter and tracking noise \cite{7553489}. These sources of misalignment collectively contribute to the vibration of the pointing direction.

\textbf{Prior Art:} 
Several studies in the literature have explored the channel model for satellite optical communications. In \cite{barry1985beam} and \cite{20099}, beam pointing errors were introduced and modeled using Gaussian distributions for both elevation and horizontal directions. The authors of \cite{arnon1997laser} examined the effects of vibrations on the bit error rate in satellite optical communications. However, these works did not consider how the beam waist changes with propagation distance, reducing the detector's received energy. Additionally, they assumed that photons radiate omnidirectionally in the channel model, overlooking the directional characteristics of lasers.
In \cite{4267802}, the combined effects of air turbulence and jitter on outage probabilities in terrestrial free-space optical (FSO) links were investigated. Another study in \cite{5955149} analyzed the impact of Hoyt-distributed pointing errors on the error performance of on-off keying optical signals. Nonetheless, these studies focused on terrestrial FSO links that accounted for atmospheric turbulence and pointing errors.
The channel statistics can be revisited in more manageable forms in the context of OISL. Moreover, the existing literature should address the average achievable data rate for OISL using laser beams across various applications. Therefore, developing an analytical model for OISL is essential to fully characterize its overall communication performance in NTN.

\textbf{Contributions:} 
The main contributions of this paper are summarized as follows.
\begin{itemize} 
\item \textit{OISL Channel Model:} We establish a channel model for OISL based on a Gaussian beam. The statistical characteristics of the OISL channel are derived by considering the effects of pointing errors and the detector's sensitivity threshold. Additionally, we derive the simplified forms of the probability density function (PDF), the cumulative distribution function (CDF), and the average channel state for the OISL channel, using Rayleigh-distributed radial deviation. Furthermore, we determine the maximum radial deviation distance associated with a Gaussian beam.

\item \textit{Cooperative OISLs Design:} We designed a cooperative OISL communication system where satellites forward data between source and destination satellites by using OISLs.
With the developed OISL channel model, we can accurately derive the average achievable data rate of an OISL communication, facilitating performance analysis of a cooperative OISL communication system.
Moreover, we introduce optimizing the number of OISL relaying satellites and laser frequency to guarantee a latency constraint and a targeted data transmission size.

\item \textit{OISLs Design Insights:} Regarding the average achievable data rate, adjusting the laser frequency or beam waist introduces a trade-off between concentrating beam energy and balancing misalignment. A reduced beam waist could potentially increase the detector's received power intensity.
However, reducing the beam waist causes the detector to deviate from the beam center and decrease the received power intensity.
Furthermore, changing the number of satellites in a cooperative OISL communication system results in a trade-off between the increased data rate of each OISL link and the sum of latency over all cooperative links.
\end{itemize}

The remainder of this paper is organized as follows.
Section \ref{system model} provides a system model. Section \ref{channel statistics} introduces the channel statistics of the OISL. Section \ref{optimization} demonstrates the achievable data rate and the total latency. Simulation results are given in Section \ref{simulations}, and the paper is concluded in Section \ref{conclusions}.

\captionsetup{font={scriptsize}}
\begin{figure*}[t!]
\begin{center}
\setlength{\abovecaptionskip}{+0.2cm}
\setlength{\belowcaptionskip}{-0.5cm}
\centering
  \includegraphics[width=6.26in, height=4.7in]{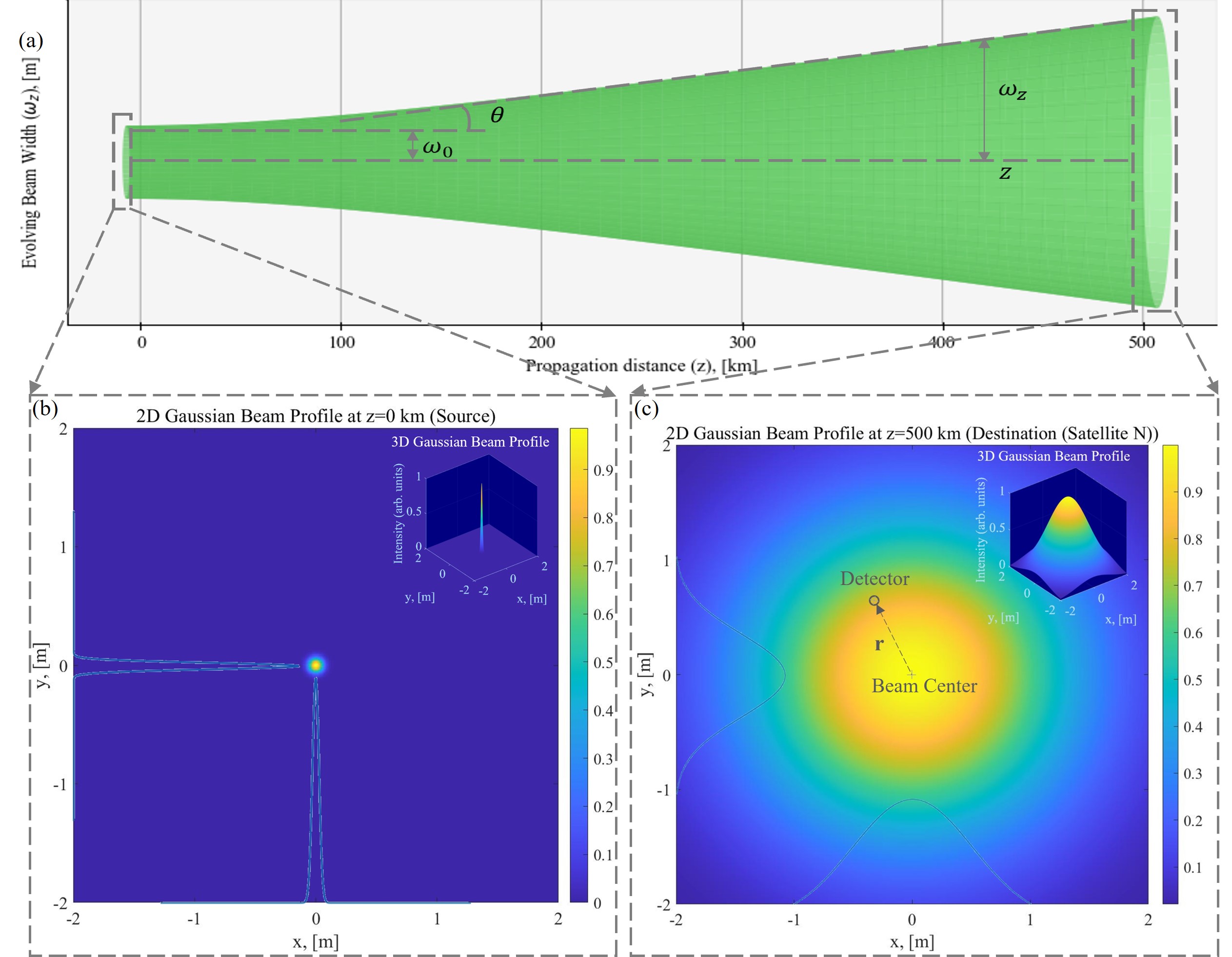}
\renewcommand\figurename{Fig}
\caption{\scriptsize An illustration of the laser beam in OISL, where $w_0=w_d=0.1$ m, the frequency of infrared light is 200 THz.}
\label{fig1}
\end{center}
\end{figure*}

\vspace{-2mm}
\section{System Model}
\label{system model}
In OISL, the received signal power $y_s$ at the detector suffers from a fluctuation in signal power, which is modeled by
\begin{equation}
\begin{aligned}
    {y_s} = h{\eta}{x_s} + {n_0}, \label{II-1}
\end{aligned}
\end{equation}
where $h$ denotes the channel state due to the expanded beam, path loss, and pointing errors, $\eta$ is the detector responsivity, ${x_s}$ is the power of the transmitted signal, and $n_0$ is signal-independent additive white Gaussian noise with variance ${\sigma _n}^2$. 
Specifically, $h$ depends on two factors \cite{4267802} shown as follows
\begin{equation}
\begin{aligned}
    h = {h_{PL}}{h_{PE}},
\end{aligned}
\end{equation}
where $h_{PL}$ encompasses the deterministic path loss which approaches one in space, since the laser beam does not propagate omnidirectionally and space is in a vacuum state \cite{hemmati2006deep}, and $h_{PE}$ indicates the channel random attenuation caused by pointing errors.
This model is suitable for any satellite or spacecraft with OISL devices.
We assume that the laser beam can be tracked and pointed with misalignment \cite{hemmati2006deep}, and the Doppler shift due to satellite movement can be well compensated.

In Fig.1, the beam waist $w_z$ at distance $z$ is given by
\begin{equation}
\begin{aligned}
    {w_z} \approx z\tan \theta  + {w_0} ,
\end{aligned}
\vspace{-2mm}
\end{equation}
where $\theta  = \frac{\lambda }{{\pi {w_0}}}$ is the beam divergence angle \cite{pampaloni2004gaussian}, $\lambda$ denotes the wavelength, and $w_0$ is the beam waist at $z=0$ m. 
Considering the long propagation distance and diffused beam, we approximate $w_z$ as follows
\begin{equation}
\begin{aligned}
    {w_z} \approx z\tan \theta . \label{II-3}
\end{aligned}
\vspace{-2mm}
\end{equation}

In Fig. 1, ${\bf{r}}$ denotes the radial vector from the beam center, i.e., pointing error, and we assume a large field of view at the receiver with no angular fluctuation \cite{8030346}.
The power intensity distribution $I\left( {{\bf{r}},z} \right)$ of a Gaussian beam \cite{4267802} at distance $z$ is
\begin{equation}
\begin{aligned}
    I\left( {{\bf{r}},z} \right) = \frac{2}{{\pi {w_z}^2}}\exp \left( { - \frac{{2{{\left\| {\bf{r}} \right\|}^2}}}{{{w_z}^2}}} \right),0 \le \left\| {\bf{r}} \right\|. \label{II-4}
\end{aligned}
\vspace{-2mm}
\end{equation}
Since the shapes of the transverse plane of the beam and detector are symmetrical, $h_{PE}$ depends on the radial deviation distance of the pointing error, i.e., $r = \left\| {\bf{r}} \right\|$, which is given by
\begin{equation}
\begin{aligned}
    {h_{PE}}\left( {r,z} \right) &= \int\limits_{{\bf{r}} \in {{\cal A}_d}} {I\left( {{\bf{r}},z} \right)} d{\bf{r}}\\
    &= \int_{ - {w_z}}^{{w_z}} {\int_0^{\sqrt {{w_d}^2 - {x^2}} } {\frac{4}{{\pi {w_z}^2}}{e^{ - 2\frac{{{{\left( {x - r} \right)}^2} + {y^2}}}{{{w_z}^2}}}}} dy} dx , \label{II-5}
\end{aligned}
\vspace{-2mm}
\end{equation}
where ${{\cal A}_d}$ is the area of the detector, $w_d$ is the radius of the detector, and we have ${{\cal A}_d} = \pi {w_d}^2$.

The main notations are summarized in Table 1.

\begin{table}[t]
\setlength{\abovecaptionskip}{0cm}
\captionsetup{font={scriptsize}}
\caption{Notations and Description}
\label{notation}
\begin{center}
\small
\begin{tabular}{c|p{6.4cm}}
\hline
\hline
\textbf{Notations } &  \textbf{Description} \\ 
\hline
\hline
$h$ & Channel power gain in the OISL \\
$n_0$ & Additive white Gaussian noise with variance ${\sigma _n}^2$ \\
$h_{PE}$ & Channel random attenuation with pointing errors \\
$w_z$ & Beam waist of the OISL at a distance $z$ \\
$\theta$ &  Beam divergence angle \\
$r$ & The radial deviation distance from the beam center \\
$w_0$ & Beam waist at the transmitter \\
$w_d$ & Radius of the detector \\
${\sigma _s}^2$ & The variance of pointing errors \\
$p_{th}$ & Sensitivity threshold at the detector \\
$h_{th}$ & Channel state average threshold at the detector \\
${r_{\max }}$ & The maximum radial deviation distance \\
${\overline h _{PE}}$ & The average channel state caused by pointing error \\
$D$ & Communication data size of OISL \\
$L$ & Communication distance of OISL\\
$N$ & The number of satellites in cooperation \\
$T_{th}$ & Total latency threshold \\
${\cal R}_n$ & Average data rate of the n-th OISL \\
\hline
\hline
\end{tabular}
\vspace{-2.0em} 
\end{center}
\end{table}

\section{OISL Channel Statistics}
\label{channel statistics}
In this section, we derive the OISL channel statistical characteristics.

\begin{proposition}
    In OISL, we have $\frac{{{w_z}}}{{{w_d}}} \gg 1$. The channel state caused by the pointing error is given by
    \begin{equation}
    \begin{aligned}
        {h_{PE}}\left( {r,z} \right) = \frac{{2{w_d}^2}}{{{w_z}^2}}\exp \left( { - \frac{{2{r^2}}}{{{w_z}^2}}} \right),{\rm{ }}0 \le r , \label{III-1}
    \end{aligned}
\vspace{-2mm}
    \end{equation}
    where $r$ is the radial deviation distance of the pointing error.
\end{proposition}
\begin{proof}
    Since the inter-satellite distances range from a few hundred kilometers to tens of thousands of kilometers, the beam waist is much larger than the radius of the detector, i.e., ${w_z} \gg {w_d}$.
    Then, we approximate that the intensity of received photon is equal across the detector, and we have ${h_{PE}}\left( {r,z} \right) \approx \pi {w_d}^2I\left( {r,z} \right)$, where $I\left( {r,z} \right)$ is given in (\ref{II-4}).
\end{proof}

\begin{theorem}
    The PDF of $h_{PE}$ at distance $z$ is given by
    \begin{equation}
    \begin{aligned}
        {f_{{h_{PE}}}}\left( y \right) &= \frac{{{w_z}}}{{2y\sqrt {2\ln \left( {\frac{{2{w_d}^2}}{{{w_z}^2y}}} \right)} }}\\
        &\cdot {f_R}\left( {\sqrt { - \frac{{{w_z}^2}}{2}\ln \left( {\frac{{{w_z}^2y}}{{2{w_d}^2}}} \right)} } \right),{\rm{ }} 0 < y \le {A_0} , \label{III-3}
    \end{aligned}
    \end{equation}
    where ${f_R}\left( r \right)$ is the PDF of radial deviation distance $r$, ${A_0} = \frac{{2{w_d}^2}}{{{w_z}^2}}$ is the maximum collected power without pointing error.
    \begin{proof}
        Let $Y = {h_{PE}}\left( {r,z} \right)$, and we have
        \begin{equation}
        \begin{aligned}
            r = g\left( y \right) = \sqrt { - \frac{{{w_z}^2}}{2}\ln \left( {\frac{{{w_z}^2y}}{{2{w_d}^2}}} \right)}  . \label{III-4}
        \end{aligned}
        \end{equation}
        Furthermore, according to \cite{devore2000probability}, we have 
        \begin{equation}
        \begin{aligned}
            {f_{{h_{PE}}}}\left( y \right) = \left\{ {\begin{array}{*{20}{l}}
            {{f_R}\left( {g\left( y \right)} \right) \cdot \left| {\frac{{dg\left( y \right)}}{{dy}}} \right|,{\rm{ }} 0 < y \le {A_0} }\\
            {0,{\rm{ otherwise}}}
            \end{array}} \right.  , \label{III-5}
        \end{aligned}
        \end{equation}
        where 
        \begin{equation}
        \begin{aligned}
            \frac{{dg\left( y \right)}}{{dy}} = \frac{{ - {w_z}}}{{2y\sqrt { - 2\ln \left( {\frac{{{w_z}^2y}}{{2{w_d}^2}}} \right)} }}  . \label{III-6}
        \end{aligned}
        \end{equation}
        By substituting (\ref{III-4}) and (\ref{III-6}) into (\ref{III-5}), we obtain (\ref{III-3}).
    \end{proof}
\end{theorem}

Without loss of generality, we reasonably assume that the horizontal and vertical misalignment, stemming from the satellite's jitter and tracking noise, follow independent and identical normal distributions \cite{barry1985beam,20099,arnon1997laser,4267802}.
Since the sum of two independent normal variables is also normally distributed, we use a single normal distribution to represent the vertical or horizontal misalignment for notation simplicity.
\begin{proposition}
    The radial deviation distance $r$ follows a Rayleigh distribution and its PDF is given as follows
    \begin{equation}
    \begin{aligned}
        {f_R}\left( r \right) = \frac{r}{{{\sigma _s}^2}}\exp \left( { - \frac{{{r^2}}}{{2{\sigma _s}^2}}} \right),0 \le r  , \label{III-7}
    \end{aligned}
    \end{equation}
    where ${\sigma _s}^2$ is the variance of pointing errors.
    \begin{proof}
        Denote ${r_x}$ and ${r_y}$ as horizontal and vertical deviation distances, respectively.
        Thus, $r = \sqrt {{r_x}^2 + {r_y}^2}$ follows a Rayleigh distribution.
    \end{proof}
\end{proposition}

\begin{theorem}
    Given the beam waist $w_z$ at distance $z$, the variance of pointing errors ${\sigma _s}^2$, the detector radius $w_d$, the PDF of $h_{PE}$ is given by
    \begin{equation}
    \begin{aligned}
        {f_{{h_{PE}}}}\left( y \right) = \frac{{{w_z}^2}}{{4{\sigma _s}^2}}{\left( {\frac{{{w_z}^2}}{{2{w_d}^2}}} \right)^{\frac{{{w_z}^2}}{{4{\sigma _s}^2}}}}{y^{\frac{{{w_z}^2}}{{4{\sigma _s}^2}} - 1}},{\rm{ }} 0 < y \le {A_0}  . \label{III-8}
    \end{aligned}
    \end{equation}
\end{theorem}
\begin{proof}
    Substituting (\ref{III-7}) into (\ref{III-3}), we obtain (\ref{III-8}), which completes the proof.
    The details are omitted to save space.
\end{proof}

In Theorem 2, we observe that the PDF of $h_{PE}$, i.e., ${f_{{h_{PE}}}}\left( y \right)$, is a power function of $y$.
It's worth noting that ${f_{{h_{PE}}}}\left( y \right)$ changes with the propagation distance $z$.

Denote $p_{th}$ as the sensitivity power threshold for a detector to identify an optical signal and $h_{th}$ as the corresponding average sensitivity threshold of $h_{PE}$, which is given by
\begin{equation}
\begin{aligned}
    {h_{th}} = \frac{{{p_{th}}}}{{{h_{PL}}\eta {P_T}}}  .
\end{aligned}
\end{equation}
Thus, the maximum radial deviation distance, i.e., $r_{max}$, is calculated by ${h_{PE}}\left( {r,z} \right) = {h_{th}}$, as follows
\begin{equation}
\begin{aligned}
    {r_{\max }} = \sqrt {\frac{{{w_z}^2}}{2}\ln \left( {\frac{{2{w_d}^2}}{{{h_{th}}{w_z}^2}}} \right)}  . \label{III-2}
\end{aligned}
\end{equation}

In some cases, the average channel state is essential to evaluating the average channel condition of OISL, as given in the following Theorem.
\begin{theorem}
    The average channel state ${\overline h _{PE}}$ caused by pointing error is given by
    \begin{equation}
    \begin{aligned}
        {\overline h _{PE}}\left( z \right) = & \frac{{{w_z}^2}}{{{w_z}^2 + 4{\sigma _s}^2}}{\left( {\frac{{{w_z}^2}}{{2{w_d}^2}}} \right)^{\frac{{{w_z}^2}}{{4{\sigma _s}^2}}}}\\
        & \cdot \left[ {{{\left( {\frac{{2{w_d}^2}}{{{w_z}^2}}} \right)}^{\frac{{{w_z}^2}}{{4{\sigma _s}^2}} + 1}} - {h_{th}}^{\frac{{{w_z}^2}}{{4{\sigma _s}^2}} + 1}} \right] , \label{III-9}
    \end{aligned}
    \end{equation}
    where ${{w_z}}$ is given in (\ref{II-3}).
\end{theorem}
\begin{proof}
    Based on Proposition 1, Proposition 2 and (\ref{III-2}), ${\overline h _{PE}}$ is obtained as follows
    \begin{equation}
    \begin{aligned}
    {{\bar h}_{PE}}\left( z \right) = & \mathbb{P} \left\{ {r \le {r_{\max }}} \right\}\int_{{h_{th}}}^{{A_0}} {y \cdot {h_{PE}}\left( {\left. y \right|y \ge {h_{th}}} \right)} dy\\
    = & \int_{{h_{th}}}^{{A_0}} {y \cdot {h_{PE}}\left( y \right)} dy , \label{III-10}
    \end{aligned}
    \end{equation}
    \color{black}
    where $\mathbb{P} \left\{ {r \le {r_{\max }}} \right\}$ is the probability that the radial deviation distance $r$ is smaller than the maximum radial deviation distance $r_{max}$, shown as follows
    \begin{equation}
    \begin{aligned}
        \mathbb{P} \left\{ {r \le {r_{\max }}} \right\} = \int_0^{{r_{\max }}} {{f_R}\left( r \right)} dr = 1 - \exp \left( { - \frac{{{r_{\max }}^2}}{{2{\sigma _s}^2}}} \right)  , \label{III-11}
    \end{aligned}
    \end{equation}
    ${h_{PE}}\left( {y} \right)$ is given in (\ref{III-8}).
    By calculating (\ref{III-10}), we obtain (\ref{III-9}).
    The detailed derivation is omitted here to save space.
\end{proof}

It is worth noting that in Theorem 3, ${\overline h _{PE}}$ can be approximated by ${{\bar h}_{PE}} \approx \frac{{2{w_d}^2}}{{{w_z}^2 + 4{\sigma _s}^2}}$ when ${{h_{th}}}$ is very small.

\begin{theorem}
    Given the beam waist $w_z$ at distance $z$, the variance of pointing errors ${\sigma _s}^2$, the detector radius $w_d$, the CDF of $h_{PE}$ is given by
    \begin{equation}
    \begin{aligned}
        {F_{{h_{PE}}}}\left( y \right) = {\left( {\frac{{{w_z}^2}}{{2{w_d}^2}}} \right)^{\frac{{{w_z}^2}}{{4{\sigma _s}^2}}}}{y^{\frac{{{w_z}^2}}{{4{\sigma _s}^2}}}},{\rm{ }} 0 < y \le {A_0}  \label{III-12}
    \end{aligned}
    \end{equation}
\end{theorem}
\begin{proof}
    The CDF of $h_{PE}$ is calculated by ${F_{{h_{PE}}}}\left( y \right) = \int_{{h_{th}}}^y {{f_{{h_{PE}}}}\left( {y'} \right)} dy'$ where ${{f_{{h_{PE}}}}\left( {y'} \right)}$ is given in (\ref{III-8}).
\end{proof}

\section{Cooperative OISL System Design}
\label{optimization}
In this section, we design a cooperative OISL communication system where relaying satellites are strategically positioned to balance the trade-off between the increased data rate of each link and the overall transmission latency.
As illustrated in Fig. 2, the channel state of the single hop between the source and destination satellites can be unfavorable due to the broader beam width and significant pointing errors that arise from the long propagation distance. 
Introducing additional relaying satellites enhances the data rate of each hop; however, the increased number of transmissions through these relays may lead to an increase in total transmission latency. 
To address this, we aim to minimize the number of cooperating satellites, denoted as $N$, while ensuring that the total latency remains below a specified threshold, $T_{th}$. This is done with the given parameters of a data size $D$ and a communication distance $L$ between the source and destination satellites.
\begin{equation}
\begin{aligned}
    & \min {\rm{ }} \text{ } N\\
    & {\rm{s}}{\rm{.t}}{\rm{. }} \text{ } \sum\limits_{n = 1}^N {\frac{D}{{{{\cal R}_n}}}}  \le {T_{th}}  , \label{IV-1}
\end{aligned}
\end{equation}
where ${\cal R}_n$ is the average data rate of the n-th OISL.

\captionsetup{font={scriptsize}}
\begin{figure}[t]
\begin{center}
\setlength{\abovecaptionskip}{+0.2cm}
\setlength{\belowcaptionskip}{-0.5cm}
\centering
  \includegraphics[width=2.8in, height=2.8in]{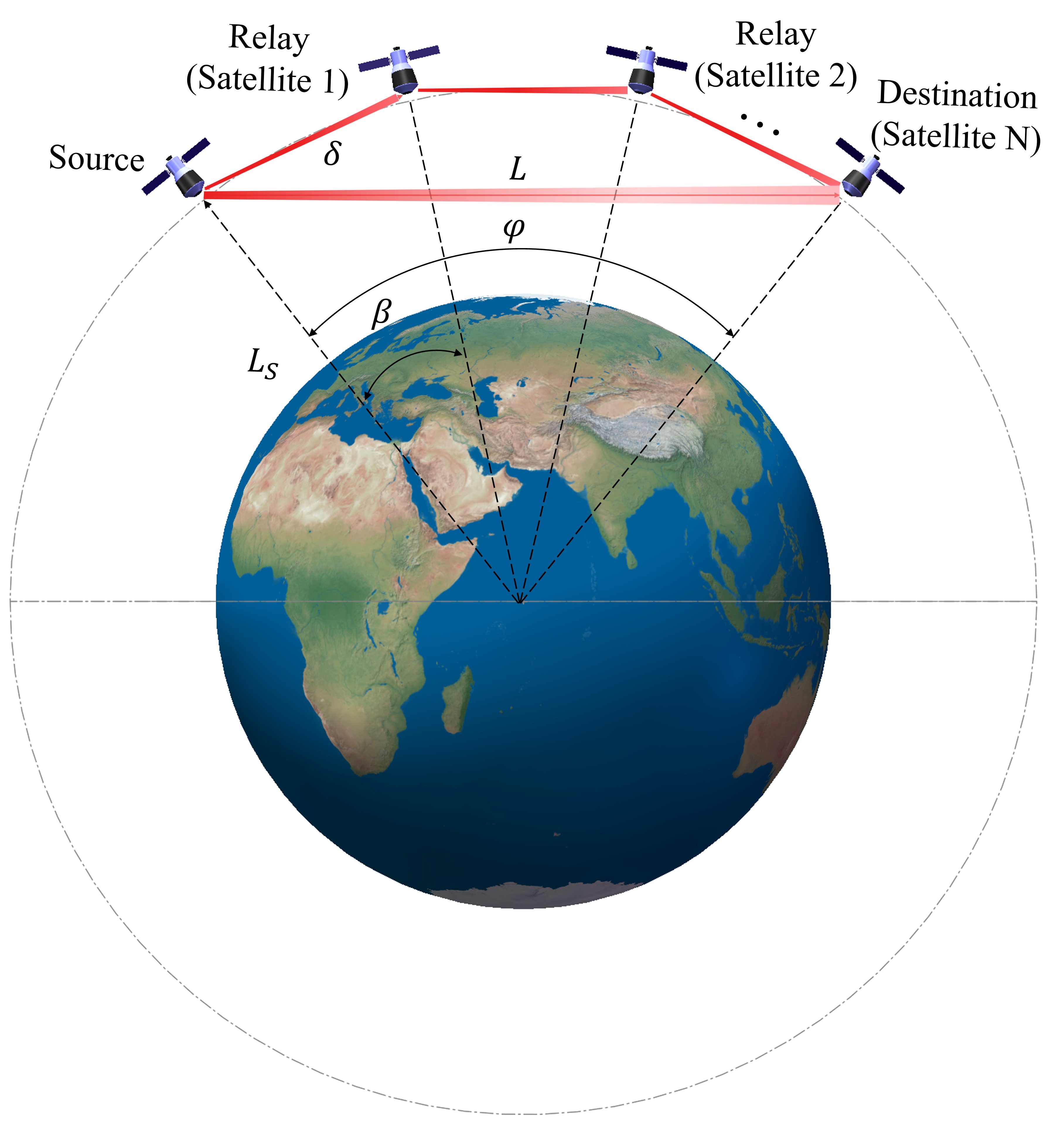}
\renewcommand\figurename{Fig}
\caption{\scriptsize An illustration of the cooperative OISLs communication system.}
\label{fig2}
\end{center}
\end{figure}

\begin{theorem}
    Given bandwidth $B$, transmit power $P_T$, communication distance $L$, the number of satellites $N$, the average achievable data rate ${\cal R}_n$ of the n-th OISL is given by
    \begin{equation}
    \begin{aligned}
        {{\cal R}_n} = \frac{{B\left( {\Upsilon \left( {{A_{0,\delta }}} \right) - \Upsilon \left( {{h_{th}}} \right)} \right)}}{{{A_{0,\delta }}^{\xi {\delta ^2}}\ln 2}}  ,  \label{IV-2}
    \end{aligned}
    \end{equation}
    where
    \begin{equation}
    \begin{aligned}
        \Upsilon \left( x \right) = & {x^{\xi {\delta ^2}}}\ln \left( {1 + {\rm{SNR}} \cdot x} \right) - \frac{{{x^{\xi {\delta ^2} + 1}}}}{{\xi {\delta ^2} + 1}}{\rm{SNR}}\\
        & \cdot {}_2{F_1}\left( {1,\xi {\delta ^2} + 1;\xi {\delta ^2} + 2; - {\rm{SNR}} \cdot x} \right) , \label{IV-3}
    \end{aligned}
    \end{equation}
    and ${A_{0,\delta }} = \frac{{{w_d}^2}}{{2{\sigma _s}^2\xi {\delta ^2}}}$ is the fraction of the collected power at beam center with $\delta$ propagation distance, $\delta  = 2{L_S}\sin \left( {\frac{1}{N}\arcsin \left( {\frac{L}{{2{L_S}}}} \right)} \right)$, $L_S$ is the distance from satellite's orbit to the center of the earth, $\xi  = \frac{{{{\tan }^2}\theta }}{{4{\sigma _s}^2}}$, ${\rm{SNR}} = \frac{{{h_{PL}}{\eta}{P_T}}}{{{\sigma _n}^2}}$, ${}_2{F_1}\left( { \cdot } \right)$ is the hypergeometric function.
\end{theorem}
\begin{proof}
    As shown in Fig. \ref{fig2}, we have $\sin \left( {\frac{\varphi }{2}} \right) = \frac{L}{{2{L_S}}}$, $\varphi  = 2\arcsin \left( {\frac{L}{{2{L_S}}}} \right)$, and $\beta  = \frac{\varphi }{N} = \frac{2}{N}\arcsin \left( {\frac{L}{{2{L_S}}}} \right)$.
    Denote the distance of each hop by $\delta$, and we have
    \begin{equation}
    \begin{aligned}
        \frac{\delta }{2} & = {L_S}\sin \left( {\frac{\beta }{2}} \right) \\
        & = {L_S}\sin \left( {\frac{1}{N}\arcsin \left( {\frac{L}{{2{L_S}}}} \right)} \right). \label{IV-Theo-1}
    \end{aligned}
    \end{equation}
    Therefore, according to (\ref{IV-Theo-1}), we have
    \begin{equation}
    \begin{aligned}
        \delta  = 2{L_S}\sin \left( {\frac{1}{N}\arcsin \left( {\frac{L}{{2{L_S}}}} \right)} \right). \label{IV-Theo-2}
    \end{aligned}
    \end{equation}
    
    The average achievable data rate of the n-th OISL is 
    \begin{equation}
    \begin{aligned}
        & {{\cal R}_n} = \mathbb{E} \left\{ {B{{\log }_2}\left( {1 + {\rm{SNR}} \cdot {h_{PE}}} \right)} \right\}\\
        & = \mathbb{P} \left\{ {{r_\delta } \le {r_{\delta ,\max }}} \right\} \mathbb{E} \left\{ {\left. {B{{\log }_2}\left( {1 + {\rm{SNR}}{h_{PE}}} \right)} \right|{h_{PE}} \ge {h_{th}}} \right\}\\
        & + \mathbb{P} \left\{ {{r_\delta } > {r_{\delta ,\max }}} \right\} \mathbb{E} \left\{ {\left. {B{{\log }_2}\left( 1 \right)} \right|{h_{PE}} < {h_{th}}} \right\}, \label{IV-Theo-3}
    \end{aligned}
    \end{equation}
    In addition, we have 
    \begin{equation}
    \begin{aligned}
        & \mathbb{P} \left\{ {{r_\delta } \le {r_{\delta ,\max }}} \right\} \mathbb{E} \left\{ {\left. {B{{\log }_2}\left( {1 +  {\rm{SNR}}} \cdot {h_{PE}} \right)} \right|{h_{PE}} \ge {h_{th}}} \right\}\\
        & = B\int_{{h_{th}}}^{{A_{0,\delta }}} {{{\log }_2}\left( {1 + {\rm{SNR}} \cdot {h_{PE}}} \right){f_{{h_{PE}}}}\left( y \right)} dy \\
        & = B\int_{{h_{th}}}^{{A_{0,\delta }}} {{{\log }_2}\left( {1 + \frac{{{h_{PL}}{R_d}{P_T}}}{{{\sigma _n}^2}}y} \right)} \\
        & \text{ }\text{ }\text{ } \cdot \xi {\delta ^2}{\left( {\frac{{2{\sigma _s}^2\xi {\delta ^2}}}{{{w_d}^2}}} \right)^{\xi {\delta ^2}}}{y^{\xi {\delta ^2} - 1}}dy\\
        & = \frac{{B\xi {\delta ^2}}}{{\ln 2}}{\left( {\frac{{2{\sigma _s}^2\xi {\delta ^2}}}{{{w_d}^2}}} \right)^{\xi {\delta ^2}}}\underbrace {\int_{{h_{th}}}^{{A_{0,\delta }}} {\frac{{\ln \left( {1 + {\rm{SNR}} \cdot y} \right)}}{{{y^{1 - \xi {\delta ^2}}}}}} dy}_\Omega . \label{IV-Theo-5}
    \end{aligned}
    \end{equation}
    Let us denote $\Omega$ as the integral in the last step of (\ref{IV-Theo-5}) as follows
    \begin{equation}
    \begin{aligned}
        & \Omega = \int_{{h_{th}}}^{{A_{0,\delta }}} {\frac{{\ln \left( {1 + {\rm{SNR}} \cdot y} \right)}}{{{y^{1 - \xi {\delta ^2}}}}}} dy. \label{IV-Theo-6}
    \end{aligned}
    \end{equation}
    Note that $\Omega$ can be further derived as follows
    \begin{equation}
    \begin{aligned}
        \Omega  & = {\Omega _1} - {\Omega _2}, \label{IV-Theo-7}
    \end{aligned}
    \end{equation}
    where
    \begin{equation}
    \begin{aligned}
        {\Omega _1}  = \int_0^{{A_0}} {\ln \left( {1 + {\rm{SNR}} \cdot y} \right){y^{\xi {\delta ^2} - 1}}} dy , \label{IV-Theo-8}
    \end{aligned}
    \end{equation}
    and 
    \begin{equation}
    \begin{aligned}
        {\Omega _2}  = \int_0^{{h_{th}}} {\ln \left( {1 + {\rm{SNR}} \cdot y} \right){y^{\xi {\delta ^2} - 1}}} dy . \label{IV-Theo-9}
    \end{aligned}
    \end{equation}
    Specifically, ${\Omega _1}$ is given in (\ref{IV-Theo-10}) at the top of the page.
    \begin{figure*}[ht] 
    \centering 
    \begin{align}
        {\Omega _1} & = \int_0^{{A_0}} {\ln \left( {1 + {\rm{SNR}} \cdot y} \right){y^{\xi {\delta ^2} - 1}}} dy = \frac{1}{{\xi {\delta ^2}}}\int_0^{{A_0}} {\ln \left( {1 + {\rm{SNR}} \cdot y} \right)} d\left( {{y^{\xi {\delta ^2}}}} \right)  \nonumber \\
        & = \frac{1}{{\xi {\delta ^2}}}\left. {\ln \left( {1 + {\rm{SNR}} \cdot y} \right){y^{\xi {\delta ^2}}}} \right|_0^{{A_0}} - \frac{1}{{\xi {\delta ^2}}}\int_0^{{A_0}} {{y^{\xi {\delta ^2}}}} d\left( {\ln \left( {1 + {\rm{SNR}} \cdot y} \right)} \right)  \nonumber \\
        & = \frac{1}{{\xi {\delta ^2}}}\ln \left( {1 + {\rm{SNR}} \cdot {A_0}} \right){A_0}^{\xi {\delta ^2}} - \frac{1}{{\xi {\delta ^2}}}\int_0^{{A_0}} {{y^{\xi {\delta ^2}}}\frac{1}{{\frac{1}{{{\rm{SNR}}}} + y}}} dy  \nonumber \\
        & = \frac{1}{{\xi {\delta ^2}}}\ln \left( {1 + {\rm{SNR}} \cdot {A_0}} \right){A_0}^{\xi {\delta ^2}} - \frac{{{\rm{SNR}} \cdot }}{{\xi {\delta ^2}}}\frac{{{A_0}^{\xi {\delta ^2} + 1}}}{{\xi {\delta ^2} + 1}}{}_2{F_1}\left( {1,\xi {\delta ^2} + 1;\xi {\delta ^2} + 2; - {\rm{SNR}} \cdot {A_0}} \right)  \nonumber \\
        & = \frac{{{A_0}^{\xi {\delta ^2}}}}{{\xi {\delta ^2}}}\left( {\ln \left( {1 + {\rm{SNR}} \cdot {A_0}} \right)} \right. - \frac{{{\rm{SNR}} \cdot {A_0}}}{{\xi {\delta ^2} + 1}}\left. {{}_2{F_1}\left( {1,\xi {\delta ^2} + 1;\xi {\delta ^2} + 2; - {\rm{SNR}} \cdot {A_0}} \right)} \right) . \label{IV-Theo-10}
    \end{align}
    \hrulefill
    \vspace{-2mm}
    \end{figure*}
    Considering $\int_0^u {{x^{\mu  - 1}}\frac{1}{{{{\left( {1 + \beta x} \right)}^\nu }}}} dx = \frac{{{u^\mu }}}{\mu }{}_2{F_1}\left( {\nu ,\mu ;1 + \mu ; - \beta u} \right)$ based on \cite{zwillinger2014table}.
    Similarly, for ${\Omega _2}$, we have
    \begin{equation}
    \begin{aligned}
        {\Omega _2} & = \frac{{{h_{th}}^{\xi {\delta ^2}}}}{{\xi {\delta ^2}}}\left( {\ln \left( {1 + {\rm{SNR}} \cdot {h_{th}}} \right)} \right.\\
        & - \frac{{{\rm{SNR}} \cdot {h_{th}}}}{{\xi {\delta ^2} + 1}}\left. {{}_2{F_1}\left( {1,\xi {\delta ^2} + 1;\xi {\delta ^2} + 2; - {\rm{SNR}} \cdot {h_{th}}} \right)} \right) .
    \end{aligned}
    \end{equation}
    We obtain the desired results with mathematical manipulations and complete the proof.
\end{proof}

In practice, the variance of pointing errors ${\sigma _s}^2$ increases with the propagation distance due to the increased difficulty of tracking.
Therefore, ${\sigma _s}$ is a function of propagation distance.
For example, we model ${\sigma _s} = {\sigma _{s,0}}{e^{{k_0}\frac{\delta }{{{d_0}}}}}$ and ${\sigma _{s,0}} = 2$, $k_0 = 0.1$, and $d_0 = 100$ km as the reference distance.

The optimal $N$ is obtained by solving the equation $ND = {T_{th}}{{\cal R}_n}$, where we consider equally spaced relay satellites.

\vspace{-2mm}
\section{Simulations and Discussions}
\label{simulations}
In this section, we conduct simulations to evaluate the OISLs' communication performance.
Unless specified otherwise, the default parameters are shown in Table I.

\begin{table}[t]
\setlength{\abovecaptionskip}{0cm}
\captionsetup{font={normalsize}}
\caption{Default Parameters Setup}
\label{Table2}
\begin{center}
\small
 \begin{tabular}{p{3.5cm} c c}
 \hline
 \hline
 \textbf{Description} & \textbf{Parameter} &  \textbf{Value} \\ 
 \hline
 \hline
 Beam waist at transmitter & $w_0$ &  0.1 \text{m} \cite{20099} \\
 Radius of the detector & $w_d$ &  0.1 \text{m} \cite{20099} \\
 Laser frequency & $f$ &  200 \text{THz} \\
 Deterministic path loss &  $h_{PL}$ & 0.9 \\
 Detector responsivity & $\eta$  &  0.5 \cite{4267802} \\
 Transmit power & $P_T$ & 0.5 Watt \\
 Bandwidth  & $B$ & 10 GHz \\
 Variance of additive noise & ${\sigma _n}^2$  & $1 \times {10^{ - 12}}$ \\
 Satellite distance to center & $L_S$& 6900 km \\
 Data size & $D$  & 100 Gbits \\
 Sensitivity threshold &  $h_{th}$ & $1 \times {10^{ - 6}}$ Watt \cite{4267802} \\
 \hline
 \hline
\end{tabular}
\vspace{-2.0em}
\end{center}
\end{table}

In Fig. 3, we compare the average channel state caused by pointing errors, i.e., ${\overline h _{PE}}$, as derived in Theorem 3, versus OISL laser frequency, i.e., $f$, under various values of $\sigma_s^2$ and different propagation distances $z$.
It is observed that ${\overline h _{PE}}$ gradually increases as the laser frequency $f$ increases.
This is because a higher frequency allows the laser beam to concentrate more energy on the detector.
However, when $\sigma_s^2$ is large, the increase in $f$ does not significantly enhance ${\overline h _{PE}}$.
This limitation is due to the misalignment caused by pointing errors, which leads the detector to deviate from the concentrated beam.
Additionally, we observe that ${\overline h _{PE}}$ decreases with the increment of $\sigma_s^2$ and $z$, which aligns with our expectations.
Moreover, the results from Monte Carlo simulations validate the analytical findings.

\captionsetup{font={scriptsize}}
\begin{figure}[t]
\begin{center}
\setlength{\abovecaptionskip}{-0.2cm}
\setlength{\belowcaptionskip}{-0.7cm}
\centering
  \includegraphics[width=3.2in, height=2.8in]{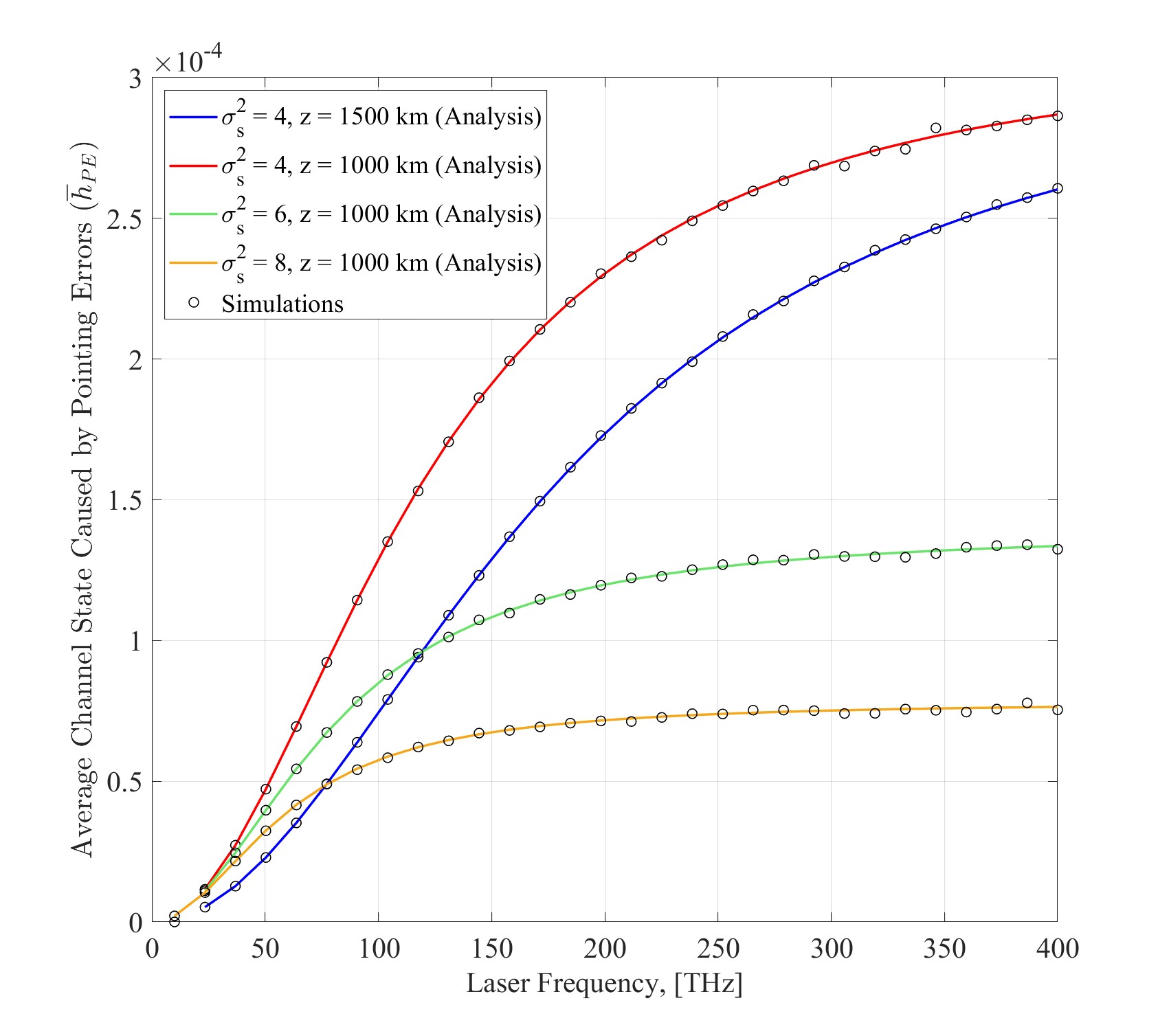}
\renewcommand\figurename{Fig.}
\caption{\scriptsize Average channel state caused by pointing errors versus laser frequency.}
\label{Fig 4}
\end{center}
\end{figure}

\captionsetup{font={scriptsize}}
\begin{figure}[t]
\begin{center}
\setlength{\abovecaptionskip}{-0.2cm}
\setlength{\belowcaptionskip}{-0.7cm}
\centering
  \includegraphics[width=3.2in, height=2.8in]{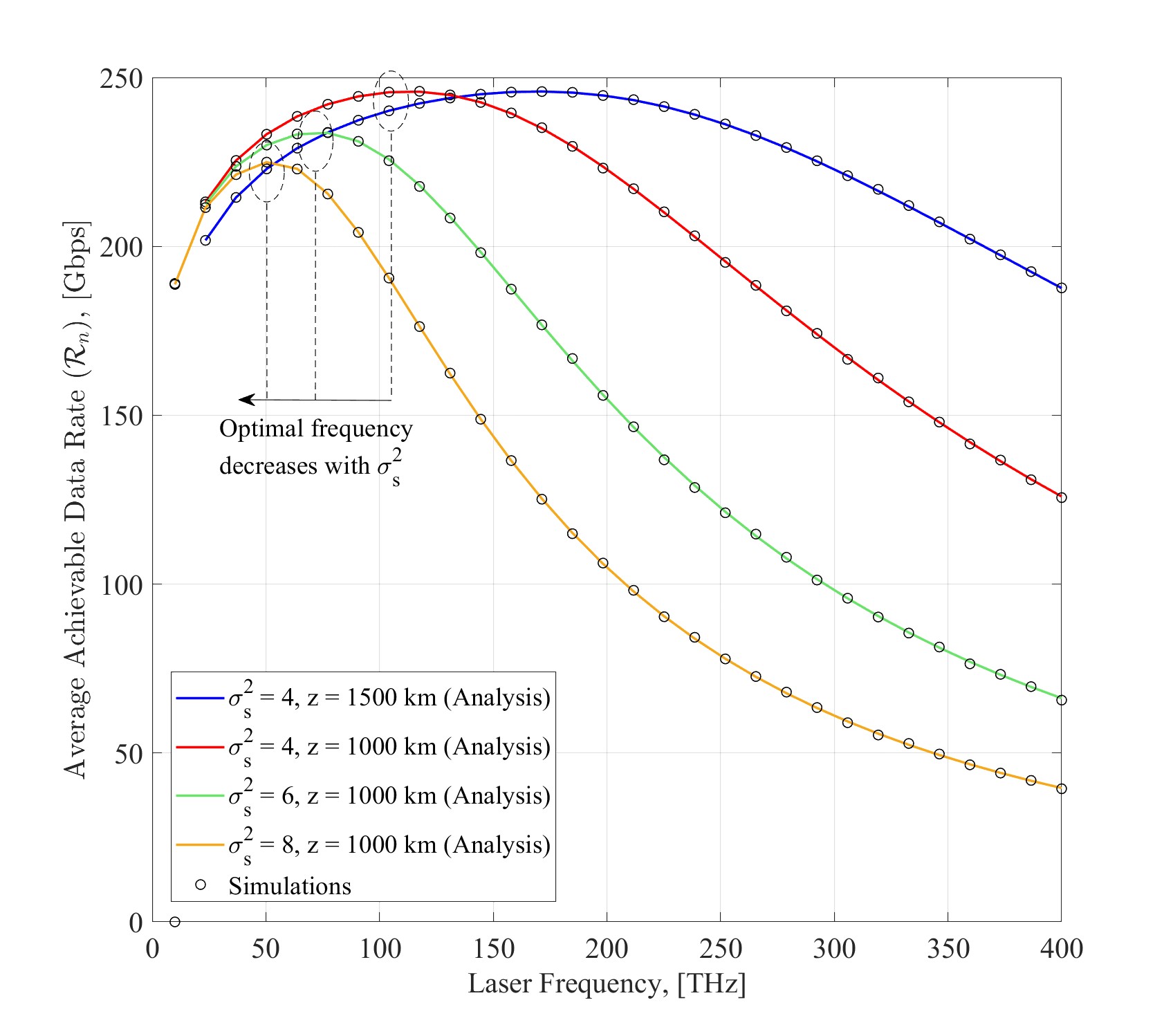}
\renewcommand\figurename{Fig.}
\caption{\scriptsize Average achievable data rate versus laser frequency.}
\label{Fig 4}
\end{center}
\end{figure}

In Fig. 4, we examine the average achievable data rate, i.e., ${{\cal R}_n}$, derived in Theorem 5 versus OISL laser frequency, i.e., $f$, under various $\sigma_s^2$ and propagation distance $z$.
We observe that there is an optimal laser frequency that maximizes ${{\cal R}_n}$.
On the one hand, when $f$ is too small, the laser's energy is diffused, and the power intensity received at the detector decreases.
On the other hand, when $f$ is too large, the detector may deviate much from the beam center due to pointing errors. In turn, the amount of energy received at the detector also decreases.
Therefore, there is a trade-off between concentrating the beam energy and balancing the misalignment.

In Fig. 5, the total latency of the cooperative OISLs system is presented concerning the number of cooperative satellites with/without frequency optimization.
The optimal laser frequency solution is calculated by exhaustive searching.
Optimizing the number of cooperative satellites and $f$ can significantly reduce the total latency.
The reduction in total latency results from a trade-off involving balancing the increased data rate of each cooperative link with the sum of latency over all cooperative links.

\captionsetup{font={scriptsize}}
\begin{figure}[t]
\begin{center}
\setlength{\abovecaptionskip}{-0.2cm}
\setlength{\belowcaptionskip}{-0.8cm}
\centering
  \includegraphics[width=3.2in, height=2.8in]{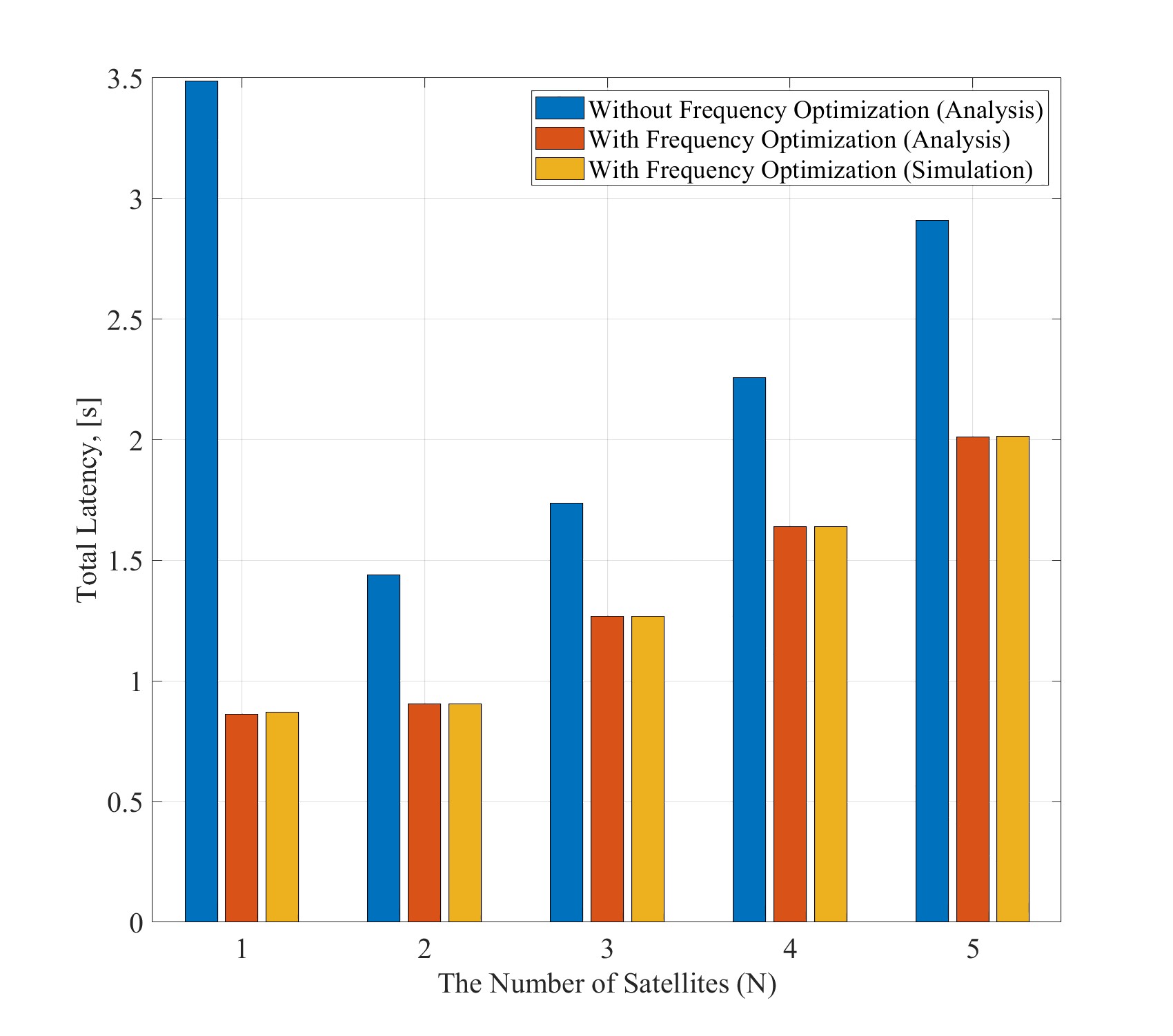}
\renewcommand\figurename{Fig.}
\caption{\scriptsize Total latency versus the number of cooperative satellites, where $L=3000$ km, frequency is within $\left[ {50,400} \right]$ THz, and frequency without optimization is at 200 THz.}
\label{Fig 5}
\end{center}
\end{figure}

\vspace{-2mm}
\section{Conclusions}
\label{conclusions}
In this paper, we developed an OISL channel model and derived its channel statistics.
Based on the channel model, we introduced a cooperative OISL system in which the number of cooperative satellites involved and laser frequency can be optimized to minimize the total latency.
Simulation results reveal the trade-offs in the design of the OISL system.
Note that the proposed analytical model applies to both current and more complicated OISL systems.
Future research can utilize the proposed model to assess the data rate and reliability of OISLs in more complex satellite mega-constellations, considering a pre-defined SINR threshold.
This paper examines OISL channel modeling while accounting for pointing errors. Future investigations could explore additional factors impacting OISL, such as satellite orbit perturbations, beam tracking, and angular fluctuations.

\ifCLASSOPTIONcaptionsoff
  \newpage
\fi
\vspace{-2mm}
\bibliographystyle{IEEEtran}
\bibliography{reference}

\end{document}